\documentclass[conference,10pt]{IEEEtran}

\IEEEoverridecommandlockouts

\usepackage{cite}
\usepackage{amsmath,amssymb,amsfonts}
\usepackage{algorithmic}
\usepackage{algorithm}
\usepackage{graphicx}
\usepackage{textcomp}
\usepackage{bm}
\usepackage{amsmath, amsfonts, amssymb, amsbsy,nccmath} \usepackage{amsthm}
\newtheorem{theorem}{Theorem}

\newcommand{\bmA}{\mathbf A}
\newcommand{\bma}{\mathbf a}
\newcommand{\bmI}{\mathbf I}
\newcommand{\bmT}{\mathbf T}

\newcommand{\bmB}{\mathbf B}
\newcommand{\bmC}{\mathbf C}
\newcommand{\bmD}{\mathbf D}
\newcommand{\bmP}{\mathbf P}
\newcommand{\bmG}{\mathbf G}
\newcommand{\bmH}{\mathbf H}
\newcommand{\bmF}{\mathbf F}
\newcommand{\bmS}{\mathbf S}
\newcommand{\bms}{\mathbf s}
\newcommand{\bmV}{\mathbf V}
\newcommand{\bmy}{\mathbf y}
\newcommand{\bmR}{\mathbf R}
\newcommand{\bmX}{\mathbf X}
\newcommand{\bmx}{\mathbf x}
\newcommand{\bmn}{\mathbf n}
\newcommand{\bmSigma}{\mathbf \Sigma}

\usepackage{xcolor}
\def\BibTeX{{\rm B\kern-.05em{\sc i\kern-.025em b}\kern-.08em
    T\kern-.1667em\lower.7ex\hbox{E}\kern-.125emX}}
\begin{document}

\title{Hybrid Beamforming and Combining for Millimeter Wave Full Duplex Massive MIMO Interference Channel
}

\author{\IEEEauthorblockN{Chandan Kumar Sheemar and Dirk Slock}
\IEEEauthorblockA{\textit{Communications Systems Department, EURECOM, France}}
\IEEEauthorblockA{\textit{ email: \{sheemar,slock\}@eurecom.fr}}}

\maketitle

\begin{abstract}
Full Duplex (FD) communication can revolutionize wireless communications as it avoids using independent channels for bi-directional communications. This work generalizes the point-to-point FD communication in millimeter wave (mmWave) band consisting of K-pairs of massive MIMO FD nodes operating simultaneously. We present a novel joint hybrid beamforming (HYBF) and combining scheme for weighted sum-rate (WSR) maximization to enable the coexistence of massive MIMO FD links cost-efficiently. The proposed algorithm relies on alternative optimization based on the minorization-maximization method. Moreover, we present a novel SI and massive MIMO interference channel aware power allocation scheme to include the optimal power control. Simulation results show significant performance improvement compared to a traditional bidirectional fully digital half-duplex (HD) system.
\end{abstract}


Full duplex (FD) can double the spectral efficiency of a wireless communication system as it allows simultaneous transmission and reception in the same frequency band. It avoids using two independent channels for bi-directional communication by allowing more flexibility in spectrum utilization, improving data security, and reducing the air interface latency and delay issues \cite{rosson2019towards,sheemar2020receiver,sheemar2021PAPC}. Self-Interference (SI) is a major challenge to deal with to achieve an ideal FD operation, which could be around $110$~dB compared to the received signal of interest.

However, continuous advancement in the SI cancellation (SIC) techniques has made the FD operation feasible. The research towards FD communication in the mmWave band (from $30$ to $300$ GHz) has recently started, which offers much wider bandwidths and results to be a  vital resource for future wireless communications. However, shifting the potential of FD to the mmWave is much more challenging, as the signal can suffer from shadowing effects, higher Doppler spreads,  rapid channel fluctuations, intermittent connectivity and low signal-to-noise ratio (SNR).  FD MIMO nodes need to be equipped with a massive number of antennas  \cite{satyanarayana2018hybrid,rosson2019towards} to overcome the propagation challenges. Therefore, to enable FD communication in mmWave cost-efficiently, we must rely on efficient hybrid beamforming (HYBF) designs, consisting of large dimensional analog processing and low dimensional digital processing.

In \cite{han2019full}, HYBF for an FD relay assisted mmWave macro-cell scenario is investigated. In \cite{satyanarayana2018hybrid}, the authors proposed a novel HYBF and combining for a point-to-point FD communication. In \cite{huang2020learning}, a machine-learning-based HYBF for a one-way mmWave FD relay is investigated.  In \cite{sheemar2021hybrid}, transmit beamforming for two massive MIMO nodes is presented. In \cite{thomas2019multi}, HYBF for a bidirectional point-to-point OFDM FD system is available. In \cite{cai2020two}, a novel HYBF design for an FD mmWave MIMO relay is proposed. In \cite{roberts2020hybrid}, a novel HYBF design with one hybrid uplink and one hybrid downlink multi-antenna half-duplex (HD) user is proposed. In \cite{sheemar2021practical}, HYBF and combining for a multi-user MIMO uplink and downlink users is proposed. In \cite{sheemar2021massive}, a HYBF for FD integrated access and backhaul is proposed. Note that the contributions on the point-to-point mmWave massive MIMO FD are limited only to a single communication link \cite{satyanarayana2018hybrid,huang2020learning,sheemar2021hybrid,thomas2019multi}.

In this paper, we generalize the point-to-point FD communication to a $K$-pairs of massive MIMO nodes in mmWave. All the nodes are assumed to be equipped with a massive number of antennas and only a limited number of radio-frequency (RF) chains. The coexistence of multiple FD nodes in mmWave leads to an FD massive MIMO interference channel in which each node suffers from SI and interference from all the other nodes. To enable FD communication in such a challenging scenario, we present a novel HYBF and combining design based on minorization-maximization \cite{stoica2004cyclic} for weighted sum-rate (WSR) maximization. Moreover, we present an optimal power allocation scheme for FD massive MIMO nodes, which are both SI and massive MIMO interference channel aware.
Simulation results show that the proposed HYBF and combining design exhibit significant performance improvement over a fully digital massive MIMO HD communication system with only a limited number of RF chains.

In summary the contributions of our work are
\begin{enumerate}
    \item Generalization of the point-to-point massive MIMO FD communication in mmWave to $K$-pair links.
    \item Introduction of a WSR maximization problem over a massive MIMO interference channel for HYBF and combining.
    \item Optimal SI and massive MIMO interference channel aware power allocation scheme. 
\end{enumerate}

\section{SYSTEM MODEL} 
In this paper, we consider a setup of $K$-pair mmWave massive MIMO FD point-to-point communication links as shown in Figure \ref{figurella_1}.
\begin{figure}[!h]
    \centering
    \includegraphics[width=8cm,height=6cm]{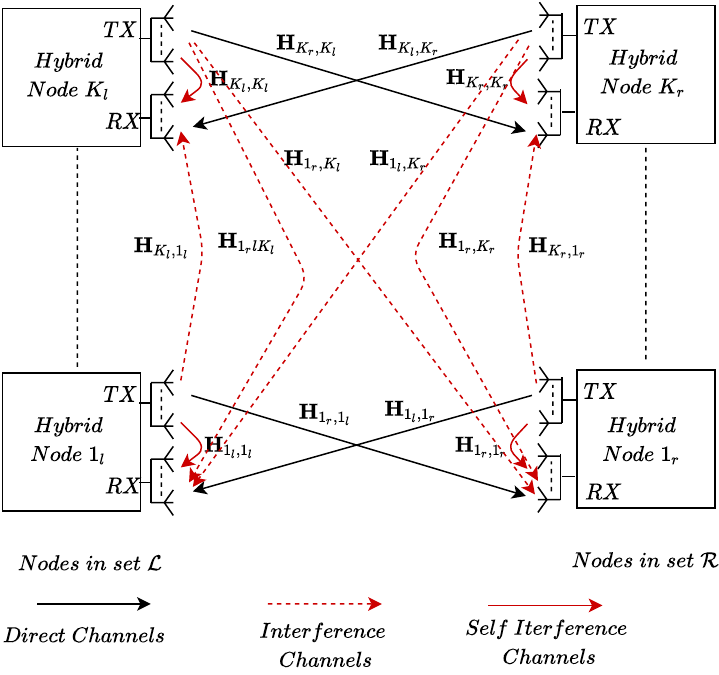}
    \caption{Massive MIMO interference channel in the mmWave consisting of Hybrid nodes with separate transmit (TX) and receive (RX) array.}
    \label{figurella_1}
\end{figure}
We distinguish the total nodes into two sets, left and right, denoted with $\mathcal{L}=\{1_l,...,K_l\}\footnote{\textbf{Notation:} Boldface lower and upper case characters denote vectors and matrices, respectively, and $\mathbb{E}\{\cdot\}, \mbox{Tr}\{\cdot\}, (\cdot)^H,\bmI$, and $\bmD_{d}$ represent expectation, trace, Hermitian transpose, identity matrix and the $d$ dominant vector selection matrix, respectively. The  operators $\mbox{vec}(\bmX),\mbox{unvec}(\bmx)$ stacks the column of $\bmX$ into a vector $\bmx$, stacks back the vector $\bmx$ into a matrix $\bmX$ and $\angle \bmX$ normalize the amplitude to be unit-modulus.
}.$ and $\mathcal{R}=\{1_r,...,K_r\}$, respectively. Nodes involved in the communication link $1\leq i \leq K$ are denoted with  $i_a$, in which the subscript $a \in \mathcal{L}$ or $\in \mathcal{R}$.
We consider a multi-stream approach and let $\bms_{i_l} \in \mathbb{C}^{d_{i_l} \times 1}$ denote the white and unitary variance data-streams transmitted from node $i_l \in \mathcal{L}$ intended for node $i_r \in \mathcal{R}$. Let $\bmV_{i_l}  \in \mathbb{C}^{M_{i_l}^t \times d_{i_l}} $ and $\bmG_{i_l} \in \mathbb{C}^{N_{i_l}^t \times M_{i_l}^t}$ denote the digital and the analog beamformer at node $i_l \in \mathcal{L}$, respectively. Let  $\bmF_{i_l} \in \mathbb{C}^{ M_{i_l}^r \times N_{i_l}^r} $ denote the analog combiner at node $i_l\in \mathcal{L}$ for the data streams $\bms_{i_r}$ transmitted from node $i_r \in \mathcal{R}$. Let $M_{i_l}^t$ and $M_{i_l}^r$ denote the number of transmit and receive RF chains for node ${i_l} \in \mathcal{L}$, respectively. Let $N_{i_l}^t$ and $N_{i_l}^r$ denote the total number of transmit and receive antennas for node $i_l \in \mathcal{L}$, respectively. The signal received  at node $i_l \in \mathcal{L}$ for the case of massive MIMO FD interference channel can be written as
\begin{equation}
\begin{aligned}
  &\bmy_{i_l} = \bmF_{i_l} \bmH_{{i_l},{i_r}} \bmG_{i_r} \bmV_{i_r} \bms_{i_r} +    \bmF_{i_l} \bmH_{{i_l},{i_l}} \bmG_{i_l} \bmV_{i_l} \bms_{i_l}  +  \bmF_{i_l} \bmn_{i_l} + \\ &   \hspace{-1mm} \sum_{\substack{m_l \in \mathcal{L}\\ m_l \neq i_l}} \hspace{-1mm}  \bmF_{i_l} \bmH_{i_l,m_l} \bmG_{m_l} \bmV_{m_l} \bms_{m_l} + \hspace{-1mm} \sum_{\substack{m_r \in \mathcal{R}\\ m_r \neq i_r}} \hspace{-1mm}  \bmF_{i_l} \bmH_{i_l,m_r} \bmG_{m_r} \bmV_{m_r} \bms_{m_r}
\end{aligned}
\end{equation}
where $\bmn_{i_l} \sim \mathcal{CN}(0,\sigma_{i_l}^2 \bmI)$ denote the noise vector at node $i$, $\bmH_{i_l,i_r} \in \mathbb{C}^{N_{i_l}^r \times N_{i_r}^t}$, and  $\bmH_{i_l,i_l} \in \mathbb{C}^{N_{i_l}^r \times N_{i_l}^t}$ denote the channel between the node $i_l \in \mathcal{L}$ and $i_r \in \mathcal{R}$, and the SI channel for node $i_l$, respectively. The matrices $\bmH_{i_l,m_l} \in \mathbb{C}^{N_{i_l}^r \times N_{m_l}^t}$ and $\bmH_{i_l,m_r} \in \mathbb{C}^{N_{i_l}^r \times N_{m_r}^t}$ denote the interference channels from node $m_l \in \mathcal{L}, m_l \neq i_l$ and from node $m_r \in \mathcal{R}, m_r \neq i_r$, respectively.
Let $\bmT_{i_l,i_r} = \bmG_{i_r} \bmV_{i_r} \bmV_{i_r}^H \bm{G}_{i_r}^H$ denote the transmit covariance matrix from node $i_r \in \mathcal{R}$ intended for node $i_l \in \mathcal{L}$. Let ($\bmR_{i_l}$) and $\bmR_{\overline{i_l}}$ denote the (signal plus) interference and noise covariance matrix received at node $i_l \in \mathcal{L}$. The covariance matrix  $\bmR_{i_l}$ can be written as
\begin{equation} \label{covariance}
    \begin{aligned}
         \bmR_{i_l}  &=  \underbrace{\bmF_{i_l} \bmH_{{i_l},{i_r}} \bmT_{i_l,i_r}  \bmH_{{i_l},{i_r}}^H  \bmF_{i_l}^H  }_{\triangleq \bmS_{i_l} } +   \bmF_{i_l} \bmH_{{i_l},{i_l}} \bmT_{i_r,i_l}
         \bmH_{{i_l},{i_l}}^H \bmF_{i_l}^H
   \\&  +  \sigma_{i_l}^2 \bmF_{i_l} \bmF_{i_l}^H  + \hspace{-3mm}    \sum_{\substack{m_l \in \mathcal{L}, m_l \neq i_l}}\hspace{-3mm}  \bmF_{i_l} \bmH_{i_l,m_l} \bmT_{m_r,m_l}   \bmH_{i_l,m_l}^H \bmF_{i_l}^H
   \\&+\hspace{-3mm} \sum_{\substack{m_r \in \mathcal{R}, m_r \neq i_r}} \hspace{-3mm} \bmF_{i_l} \bmH_{i_l,m_r} \bmT_{m_l,m_r} \bmH_{i_l,m_r}^H \bmF_{i_l}^H
    \end{aligned}
\end{equation} 
and $\bmR_{\overline{i_l}}$ can be written as $\bmR_{\overline{i_l}} = \bmR_{i_l} - \bmS_{i_l}$,
where $\bmS_{i_l}$ denotes the useful signal part for node $i_l \in \mathcal{L}$. In the mmWave, the channel $\bmH_{i_l,i_r}$ can be modelled as 

\begin{equation} 
   \bmH_{i_l,i_r} = \sqrt{\frac{N_{i_l}^r N_{i_r}^t}{N_c N_{p}}} \sum_{n_c = 1}^{N_c} \sum_{n_p = 1}^{N_p} \alpha_{i_l,i_r}^{n_p,n_c} \bma_{i_l}(\phi_{i_l}^{n_p,n_c}) \bma_{i_r}^T(\theta_{i_r}^{n_p,n_c}), \label{channel_model}
\end{equation} 
where $N_c$ and $N_{p}$ denote the number of clusters and number of rays, respectively, $\alpha_{i_l,i_r}^{n_p,n_c} \sim \mathcal{CN}(0,1)$ is a complex Gaussian random variable with amplitudes and phases distributed according to the Rayleigh and uniform distribution, respectively, and $\bma_{i_l}(\phi_{i_l}^{n_p,n_c})$ and  $\bma_{i_r}^T(\theta_{i_r}^{n_p,n_c})$ denote the receive and transmit antenna array response with angle of arrival (AoA) $\phi_{i_l}^{n_p,n_c}$ and angle of departure (AoD)   $\theta_{i_r}^{n_p,n_c}$, respectively. Also, the channel matrices $\bmH_{i_l,m_r}$ and $\bmH_{i_l,m_l}$ can be modelled as \eqref{channel_model}. The SI channel can be modelled with the Rician fading channel model given as
\begin{equation} \label{SI_Channel}
    \bmH_{i_l,i_l} = \sqrt{\frac{\kappa_{i_l}}{\kappa_{i_l}+1}} \bmH_{i_l}^{LoS} + \sqrt{\frac{1}{\kappa_{i_l}+1}} \bmH_{i_l}^{Ref},
\end{equation}
where $\kappa_{i_l}$, $\bmH_{i_l}^{LoS}$ and $\bmH_{i_l}^{Ref}$ denote the Rician factor, the line-of-sight (LoS) component matrix and the matrix of reflected components, respectively, at node $i_l \in\mathcal{L}$. The channel matrix $\bmH_{i_l}^{Ref}$ can be modelled as in \eqref{channel_model}. The elements of matrix $\bmH_{i_l}^{LoS}$ at $m$-th row and $n$-th column
can be modelled as
\begin{equation} \label{SI_LOS_model}
    \bmH_{i_l}^{LoS}(m,n) = \frac{\rho}{r_{m,n}} e^{-j 2 \pi \frac{r_{m,n}}{\lambda}}
\end{equation}
 where $\rho$ is the power normalization constant which assure that $\mathbb{E}(||\bmH_{i_l}^{LoS}(m,n)||_F^2)= N_{i_l}^r N_{i_l}^t$ and $r_{m,n}$ depends on the antenna array geometry. The WSR maximization problem for a massive MIMO interference channel under the total sum-power constraint and the unit-modulus phase shifters over the mmWave FD massive MIMO interference channel can be stated as
\begin{subequations}\label{WSR}
\begin{equation}
     \underset{\substack{\bmG,\bmV,\\ \bmF}}{\text{max}} \quad \sum_{i_l \in \mathcal{L}} w_{i_l} \mbox{ln det}(\bmR_{\overline{{i_l}}}^{-1} \bmR_{{i_l}}) + \sum_{{i_r} \in \mathcal{R}} w_{i_r} \mbox{ln det}(\bmR_{\overline{{i_r}}}^{-1} \bmR_{{i_r}}).
\end{equation}
\begin{equation}
\text{s.t.} \quad \mbox{Tr}( \bmG_a \bmV_a \bmV_a^H  \bmG_a^H ) 	\preceq p_a, \quad \forall a \in \mathcal{L} \;\mbox{or}\; \in \mathcal{R} , \label{c1}
\end{equation}
\begin{equation}
 \quad  \quad \quad |\bmG_a(m,n)|^2 = 1, \quad \forall \; m,n,  \quad \forall a \in \mathcal{L} \;\mbox{or}\; \in \mathcal{R} , \label{c2}
\end{equation} 
\begin{equation}
 \quad  \quad \quad |\bmF_a(m,n)|^2 = 1, \quad \forall \; m,n,  \quad \forall a \in \mathcal{L} \;\mbox{or}\; \in \mathcal{R}. \label{c3}
\end{equation} 
\end{subequations}
The scalars $w_{i_l}$ and $w_{i_r}$ denote the rate weights for node $i_l \in \mathcal{L}$ and $i_r \in \mathcal{R}$, respectively, and $p_a$
denote the total sum-power constraint for node $a \in \mathcal{L} \;\mbox{or}\; \in \mathcal{R}$. The constraints $\eqref{c2}$ and $\eqref{c3}$ denote the unit-modulus phase-shifter constraint on node $a$ for the analog beamformer and combiner, respectively. In the problem statement, $\bmG,\bmV,$ and $\bmF $, denote the collection of the analog and digital beamformers and the analog combiners, respectively.

\emph{Remark-} Note that \eqref{WSR}, stated as $\mbox{ln det}( \cdot )$, is not affected by the digital combiners. They can be chosen to be the well known MMSE combiners and the achieved rate would not be affected (Please see (4)-(9) \cite{christensen2008weighted}).

\section{Problem Simplification}
 The problem \eqref{WSR} is non-concave in terms of the covariance matrices for the weighted rates of nodes for which the covariance matrices generate interference. Given a non-concave structure of the problem due to interference, it makes the searching for global optima extremely challenging. 
 
 To render a feasible solution, we leverage the minorization-maximization \cite{stoica2004cyclic} approach. Note that the WSR can be written as 
\begin{equation}
    WSR = WSR_L +  WSR_R = \sum_{i_l \in \mathcal{L}} WR_{i_l} + \sum_{i_r \in \mathcal{R}} WR_{i_r} 
\end{equation}
 where $WR$ denotes the weighted-rate of one particular user, denoted with the subscript. Note that $WR_{i_l}$ is only concave in $\bmT_{i_l,i_r}$ (point-to-point link between nodes $i_l$ and $i_r$) and the remaining WRs are non-concave in $\bmT_{i_l,i_r}$.
  Let $\overline{i_l}$( $\overline{i_r}$) denote the summation in set $\mathcal{L}$ ( $\mathcal{R}$) with the element $i_l$ ( $i_r$).
Since a linear function is simultaneously convex and concave, difference of convex (DC) programming introduces the first order
Taylor series expansion of $WSR_{\overline{i_l}}$ and $WSR_{R}$ in $\bmT_{i_l,i_r}$, around $\hat{\bmT}_{i_l,i_r}$, (i.e. all $\bmT_{i_l,i_r}$). Let $\hat{\bmT}$ denote the set containing all such $\hat{\bmT}_{i_l,i_r}$. The linearized tangent expression for the FD point-to-point link for the transmit covariance matrices of node $i_l \in \mathcal{L}$ and $i_r \in \mathcal{L}$ can be written by computing the following gradients 
\begin{equation}
    \hat{\bmA}_{i_r} = - \frac{\partial \mbox{WSR}_{\overline{i_l}}}{\partial \hat{\bmT}_{i_l,i_r}}\Bigr|_{\hat{\bmT}}
    ,\quad   \hat{\bmB}_{i_r}  = - \frac{\partial \mbox{WSR}_R}{\partial \hat{\bmT}_{i_l,i_r}}\Bigr|_{\hat{\bmT}}, 
    \label{grad_UL1}
\end{equation}
 \begin{equation}
    \hat{\bmA}_{i_l} = - \frac{\partial \mbox{WSR}_{\overline{i_r}}}{\partial \hat{\bmT}_{i_r,i_l}}\Bigr|_{\hat{\bmT}}
    ,\quad   \hat{\bmB}_{i_l}  = - \frac{\partial \mbox{WSR}_L}{\partial \hat{\bmT}_{i_r,i_l}}\Bigr|_{\hat{\bmT}}. 
    \label{grad_UL2}
\end{equation}
 where $\hat{\bmA}_{a}$ and $\hat{\bmB}_{a}$, for $a \;\in\; \mathcal{L} \;\mbox{or} \;\mathcal{R}$, denote the linearization with respect to the same set and the other set, respectively. The gradients can be computed by using the matrix differentiation properties, which leads to the expressions 
 \begin{subequations} \label{grad}
  \begin{equation} \label{A_il}
      \hat{\bmA}_{i_r} = \hspace{-3mm} \sum_{\substack{m_l \in \mathcal{L}, m_l \neq i_l}} \hspace{-3mm}\bmH_{m_l,i_r}^H \bmF_{m_l}^H (\bmR_{\overline{m_l}}^{-1}) - \bmR_{m_l}^{-1}))\bmF_{m_l} \bmH_{m_l,i_r}
 \end{equation}
  \begin{equation} \label{B_il}
      \hat{\bmB}_{i_r} =  \sum_{\substack{n_r \in \mathcal{R}}} \bmH_{n_r,i_r}^H \bmF_{n_r}^H (\bmR_{\overline{n_r}}^{-1}) - \bmR_{n_r}^{-1}))\bmF_{n_r}\bmH_{n_r,i_r},
 \end{equation}   
  \begin{equation} \label{A_ir}
   \hat{\bmA}_{i_l} = \hspace{-3mm} \sum_{\substack{m_r \in \mathcal{R}, m_r \neq i_r}} \hspace{-3mm} \bmH_{m_r,i_l}^H \bmF_{m_r}^H (\bmR_{\overline{m_r}}^{-1}) - \bmR_{m_r}^{-1}))\bmF_{m_r} \bmH_{m_r,i_r},
 \end{equation}
  \begin{equation} \label{B_ir}
    \hat{\bmB}_{i_l} = \sum_{\substack{n_l \in \mathcal{L}}} \bmH_{n_l,i_l}^H \bmF_{n_l}^H (\bmR_{\overline{n_l}}^{-1}) - \bmR_{n_l}^{-1}))\bmF_{n_l}\bmH_{n_l,i_l}.
 \end{equation}
 \end{subequations}
Note that the rate of node $i_r$  depends on the transmit covariance matrix from node $i_l$  and vice-versa, and the gradients $ \hat{\bmB}_{i_r}$ and $\hat{\bmB}_{i_l}$ take into account also the SI generated at the node $i_l$ and  $i_r$, respectively. Let $\lambda_{i_l}$ and $\lambda_{i_r}$ denote the  Lagrange multipliers associated with the sum-power constraint for the node $i_l$ and $i_r$, respectively. Considering the unconstrained analog part, dropping the constant terms and re-parametrizing back the transmit covariance matrices as a function of the digital and the analog beamformers, augmenting the WSR cost function with sum-power constraint, leads to the Lagrangian  \eqref{WSR_concave}, given at the top of the next page. Note that \eqref{WSR_concave} does not contain the unit-modulus constraints, which will be incorporated later.
\begin{figure*}[t]
\begin{equation} \label{WSR_concave}
\begin{aligned}
       \mathbb{L}= & \sum_{i_l \in \mathcal{L}} \lambda_{i_l} p_{i_l} +  \sum_{i_r \in \mathcal{R}} \lambda_{i_r} p_{i_r} +  \sum_{i_l \in \mathcal{L}} w_{i_l} \mbox{ln det} (\bmI + \bmV_{i_r}^H \bmG_{i_r}^H \bmH_{i_l,i_r}^H \bmF_{i_l}^H \bmR_{\overline{i_l}}^{-1}\bmF_{i_l} \bmH_{i_l,i_r}  \bmG_{i_r}  \bmV_{i_r})  - \mbox{Tr}( \bmV_{i_r}^H  \bmG_{i_r}^H (\hat{\bmA}_{i_r} + \hat{\bmB}_{i_r} + \lambda_{i_r} \bmI) \\& \bmG_{i_r}  \bmV_{i_r})  + \sum_{i_r \in \mathcal{R}} w_{i_r} \mbox{ln det} (\bmI + \bmV_{i_l}^H \bmG_{i_l}^H \bmH_{i_r,i_l}^H \bmF_{i_r}^H \bmR_{\overline{i_r}}^{-1}\bmF_{i_r} \bmH_{i_r,i_l} \bmG_{i_l}  \bmV_{i_l}) 
           - \mbox{Tr}( \bmV_{i_l}^H \bmG_{i_l}^H (\hat{\bmA}_{i_l} + \hat{\bmB}_{i_l} + \lambda_{i_l} \bmI) \bmG_{i_l}  \bmV_{i_l} ),
\end{aligned} \vspace{-2mm}
\end{equation}\hrulefill
\end{figure*}

\section{Hybrid Beamforming and Combining} 
This sections presents a novel HYBF and combining design based on the problem simplification in \eqref{WSR_concave}.

\subsection{Digital beamformers}
 To optimize the digital for node $i_l \in \mathcal{L}$ and $i_r \in \mathcal{R}$, we take the derivative of the \eqref{WSR_concave}  with respect to $\bmV_{i_r}$ and $\bmV_{i_l}$, respectively, which leads to the following Karush–Kuhn–Tucker (KKT) conditions
 \begin{subequations}
 \begin{equation}
\begin{aligned}
 &\bmG_{i_r}^H \bmH_{i_l,i_r}^H \bmF_{i_l}^H \bmR_{\overline{i_l}}^{-1}\bmF_{i_l} \bmH_{i_l,i_r} \bmG_{i_r}  \bmV_{i_r} \big(\bmI + \bmV_{i_r}^H \bmG_{i_r}^H \bmH_{i_l,i_r}^H \bmF_{i_l}^H \\& \bmR_{\overline{i_l}}^{-1}\bmF_{i_l} \bmH_{i_l,i_r}  \bmG_{i_r}  \bmV_{i_r}\big)^{-1}   =   \bmG_{i_r}^H \big( \hat{\bmA}_{i_r} + \hat{\bmB}_{i_r}  +  \lambda_{i_r} \bmI\big)\bmG_{i_r} \bmV_{i_r}, 
\end{aligned} \label{grad_precoder_i} 
\end{equation} 
 \begin{equation}
\begin{aligned}
&\bmG_{i_l}^H \bmH_{i_r,i_l}^H \bmF_{i_r}^H \bmR_{\overline{i_r}}^{-1}\bmF_{i_r} \bmH_{i_r,i_l}  \bmG_{i_l}  \bmV_{i_l}\big(\bmI + \bmV_{i_l}^H \bmG_{i_l}^H \bmH_{i_r,i_l}^H \bmF_{i_r}^H \\& \bmR_{\overline{i_r}}^{-1}\bmF_{i_r} \bmH_{i_r,i_l}  \bmG_{i_l}  \bmV_{i_l}\big)^{-1}   =   \bmG_{i_l}^H \big( \hat{\bmA}_{i_l} + \hat{\bmB}_{i_l}  +  \lambda_{i_l} \bmI\big)\bmG_{i_l} \bmV_{i_l}, 
\end{aligned} \label{grad_precoder_j} 
\end{equation} 
\end{subequations}

\begin{theorem} \label{teorema1}
The digital beamformers $\bmV_{i_l}$ and $\bmV_{i_r}$ can be optimized as a generalized dominant eigenvectors of the pairs 
\begin{equation}
\begin{aligned} \label{digital_Vir}
          \bmV_{i_r} = \bmD_{d_{i_l}}(&\bmG_{i_r}^H \bmH_{i_l,i_r}^H \bmF_{i_l}^H \bmR_{\overline{i_l}}^{-1}\bmF_{i_l} \bmH_{i_l,i_r} \bmG_{i_r} ,  \\&\bmG_{i_r}^H \big( \hat{\bmA}_{i_r} + \hat{\bmB}_{i_r}  +  \lambda_r \bmI\big)\bmG_{i_r}) ,
\end{aligned}
\end{equation}
\begin{equation} \label{digital_Vil}
\begin{aligned}
          \bmV_{i_l} = \bmD_{d_{i_r}}(&\bmG_{i_l}^H \bmH_{i_r,i_l}^H \bmF_{i_r}^H \bmR_{\overline{i_r}}^{-1}\bmF_{i_r} \bmH_{i_r,i_l} \bmG_{i_l} ,  \\&\bmG_{i_l}^H \big( \hat{\bmA}_{i_l} + \hat{\bmB}_{i_l}  +  \lambda_{i_l} \bmI\big)\bmG_{i_l}) ,
\end{aligned}
\end{equation}
 where $\bmD_{d_{i_l}}$ $(\bmD_{d_{i_r}})$ selects the $d_{i_l}$ $(d_{i_r})$ dominant generalized eigenvectors. 
\end{theorem}
\begin{proof}
 Given the analog part fixed, we proof the result only for $\bmV_{i_l}$. The result for $\bmV_{i_r}$ is then straightforward based on the proof for $\bmV_{i_l}$.
The proof relies on simplifying \eqref{WSR_concave} with fixed analog part until the
Hadamard’s inequality applies. The Cholesky
decomposition of the matrix $\bmG_{i_l}^H \big( \hat{\bmA}_{i_l} + \hat{\bmB}_{i_l}  +  \lambda_{i_l} \bmI\big)\bmG_{i_l})$ is written as
$\bm{L}_{i_l} \bm{L}_{i_l}^H$ where $\bm{L}_{i_l}$ is a lower-triangular Cholesky factor.
Similarly we do Cholesky
decomposition also for $\bmG_{i_l}^H \bmH_{i_r,i_l}^H \bmF_{i_r}^H \bmR_{\overline{i_r}}^{-1}\bmF_{i_r} \bmH_{i_r,i_l} \bmG_{i_l} $. Given the Cholesky
decomposition, we can apply the result provided in Proposition 1 \cite{kim2011optimal}, with the matrices $\bmG_{i_l}^H \big( \hat{\bmA}_{i_l} + \hat{\bmB}_{i_l}  +  \lambda_{i_l} \bmI\big)\bmG_{i_l})$ and $\bmG_{i_l}^H \bmH_{i_r,i_l}^H \bmF_{i_r}^H \bmR_{\overline{i_r}}^{-1}\bmF_{i_r} \bmH_{i_r,i_l} \bmG_{i_l} $. It follows immediately that the digital beamformer $\bmV_{i_l}$ is given as a dominant generalized eigenvectors of these two matrices. The same reasoning follows also for the digital beamformer $\bmV_{i_r}$. 
\end{proof}
Note that the GEV solution provides the optimal beamforming directions but not the optimal power allocation. Therefore, to include the optimal power allocation, we normalize the columns of the digital beamformers to be unit norm.

\subsection{Analog beamformer}
To optimize the analog beamformers for node $i_l \in \mathcal{L}$ and $i_r \in \mathcal{R}$, we take the derivative of \eqref{WSR_concave} for $\bmG_{i_l}$ and $\bmG_{i_r}$, which leads to the following KKT conditions
\begin{subequations}
 \begin{equation}
 \begin{aligned}
 &\bmH_{i_r,i_l}^H \bmF_{i_r}^H  \bmR_{\overline{i_r}}^{-1}\bmF_{i_r} \bmH_{i_r,i_l}  \bmG_{i_l} \bmV_{i_l} \bmV_{i_l}^H  
     \big(\bmI + \bmV_{i_l}^H \bmG_{i_l}^H \bmH_{i_r,i_l}^H \bmF_{i_r}^H \\& \bmR_{\overline{i_r}}^{-1}\bmF_{i_r} \bmH_{i_r,i_l}  \bmG_{i_l}  \bmV_{i_l}\big)^{-1}  =  
      \big( \hat{\bmA}_{i_l} + \hat{\bmB}_{i_l}  +  \lambda_{i_l} \bmI\big)\bmG_{i_l} \bmV_{i_l} \bmV_{i_l}^H,
 \end{aligned}
 \end{equation}
  \begin{equation}
 \begin{aligned}
 & \bmH_{i_l,i_r}^H \bmF_{i_l}^H  \bmR_{\overline{i_l}}^{-1}\bmF_{i_l} \bmH_{i_l,i_r} \bmG_{i_r} \bmV_{i_r} \bmV_{i_r}^H
     \big(\bmI + \bmV_{i_r}^H \bmG_{i_r}^H \bmH_{i_l,i_r}^H \bmF_{i_l}^H \\& \bmR_{\overline{i_l}}^{-1}\bmF_{i_l} \bmH_{i_l,i_r}  \bmG_{i_r}  \bmV_{i_r}\big)^{-1}  =  
      \big( \hat{\bmA}_{i_r} + \hat{\bmB}_{i_r}  +  \lambda_{i_r} \bmI\big)\bmG_{i_r} \bmV_{i_r} \bmV_{i_r}^H.
 \end{aligned}
 \end{equation}
\end{subequations}
To optimize the analog beamformer, the KKT conditions are not resolveable for $\bmG_{i_l}$ and $\bmG_{i_r}$. To do so, we apply the identity $\mbox{vec}(\bmA \bmX \bmB) = (\bmB^T \otimes \bmA) \mbox{vec}(\bmX)$, which shapes the the KKT conditions as

\begin{subequations}
 \begin{equation}
\begin{aligned}
    & [( \bmV_{i_r} \bmV_{i_r}^H
     \big(\bmI + \bmV_{i_r}^H \bmG_{i_r}^H \bmH_{i_l,i_r}^H \bmF_{i_l}^H  \bmR_{\overline{i_l}}^{-1}\bmF_{i_l} \bmH_{i_l,i_r}  \bmG_{i_r}  \bmV_{i_r}\big)^{-1})^T \\& \otimes \bmH_{i_l,i_r}^H \bmF_{i_l}^H  \bmR_{\overline{i_l}}^{-1}\bmF_{i_l} \bmH_{i_l,i_r}] \mbox{vec}(\bmG_{i_r}) \\& = [(\bmV_{i_r} \bmV_{i_r}^H)^T \otimes \big( \hat{\bmA}_{i_r} + \hat{\bmB}_{i_r}  +  \lambda_{i_l} \bmI\big)] \mbox{vec}(\bmG_{i_r}),
\end{aligned}
 \end{equation}
  \begin{equation}
  \begin{aligned}
   & [( \bmV_{i_l} \bmV_{i_l}^H
     \big(\bmI + \bmV_{i_l}^H \bmG_{i_l}^H \bmH_{i_r,i_l}^H \bmF_{i_r}^H  \bmR_{\overline{i_r}}^{-1}\bmF_{i_r} \bmH_{i_r,i_l}  \bmG_{i_l}  \bmV_{i_l}\big)^{-1})^T \\& \otimes \bmH_{i_r,i_l}^H \bmF_{i_r}^H  \bmR_{\overline{i_r}}^{-1}\bmF_{i_r} \bmH_{i_r,i_l} ] \mbox{vec}(\bmG_{i_l}) \\& = [(\bmV_{i_l} \bmV_{i_l}^H)^T \otimes \big( \hat{\bmA}_{i_l} + \hat{\bmB}_{i_l}  +  \lambda_{i_r} \bmI\big)] \mbox{vec}(\bmG_{i_l}).
  \end{aligned}
 \end{equation}
\end{subequations}
Given the vectorized analog beamformer, by following a similar proof as for the digital beamformers in the Theorem \eqref{teorema1}, it can be easily seen that the unconstrained vectorized analog combiners $\mbox{vec}(\bmG_{i_r})$ and $\mbox{vec}(\bmG_{i_l})$ are given by
\begin{subequations}
\begin{equation} \label{analog_Gir}
\begin{aligned} 
    \mbox{vec}(\bmG_{i_r}) = & \bmD_1(( \bmV_{i_r} \bmV_{i_r}^H
     \big(\bmI + \bmV_{i_r}^H \bmG_{i_r}^H \bmH_{i_l,i_r}^H \bmF_{i_l}^H  \bmR_{\overline{i_l}}^{-1}\bmF_{i_l} \\& \bmH_{i_l,i_r}  \bmG_{i_r}  \bmV_{i_r}\big)^{-1})^T  \otimes \bmH_{i_l,i_r}^H \bmF_{i_l}^H  \bmR_{\overline{i_l}}^{-1}\bmF_{i_l} \bmH_{i_l,i_r} , \\& (\bmV_{i_r} \bmV_{i_r}^H)^T \otimes \big( \hat{\bmA}_{i_l} + \hat{\bmB}_{i_l}  +  \lambda_{i_r} \bmI\big)),
\end{aligned}
\end{equation}
\begin{equation} \label{analog_Gil}
\begin{aligned} 
    \mbox{vec}(\bmG_{i_l}) = & \bmD_1(( \bmV_{i_l} \bmV_{i_l}^H
     \big(\bmI + \bmV_{i_l}^H \bmG_{i_l}^H \bmH_{i_r,i_l}^H \bmF_{i_r}^H  \bmR_{\overline{i_r}}^{-1}\bmF_{i_r} \\& \bmH_{i_r,i_l}  \bmG_{i_l}  \bmV_{i_l}\big)^{-1})^T  \otimes \bmH_{i_r,i_l}^H \bmF_{i_r}^H  \bmR_{\overline{i_r}}^{-1}\bmF_{i_r} \bmH_{i_r,i_l} , \\& (\bmV_{i_r} \bmV_{i_r}^H)^T \otimes \big( \hat{\bmA}_{i_r} + \hat{\bmB}_{i_r}  +  \lambda_{i_l} \bmI\big)),
\end{aligned}
\end{equation}
\end{subequations}
which needs to be reshaped into correct dimensions with the operation $\mbox{unvec}(\cdot)$. Furthermore, to meet the unit modulus constraint, we apply the operation $\angle \bmG_{i_l}$ and $\angle \bmG_{i_r}$, which preserves only the phase part.

\subsection{Analog combiner}
To optimize the analog combiner, we first define the received covariance matrices at the antenna level as $\bmR_{i_l}^{ant}$ and $\bmR_{i_r}^{ant}$ (obtained from \eqref{covariance} by simply omitting the analog combiner). After the analog combining stage, we have the following expression for the covariance matrices $\bmR_{i_l} = \bmF_{i_l} \bmR_{i_l}^{ant} \bmF_{i_l}^H $ and $\bmR_{i_r} = \bmF_{i_r} \bmR_{i_r}^{ant} \bmF_{i_r}^H $. Note that in the minorization-maximization approach, we linearize with respect to the WRs for which each transmit covariance matrix generate interference. However, as the combiner do not generate any interference and does not have the sum-power constraint,  we can directly solve \eqref{WSR}, which result to be concave for the analog combiners. Namely, we can first write

\begin{equation} \label{WSR_restated}
\begin{aligned}
   \underset{\substack{\bmF}}{\text{max}}  &\sum_{i_l \in \mathcal{L}} w_{i_l} [\mbox{ln det}( \bmF_{i_l} \bmR_{i_l}^{ant} \bmF_{i_l}^H )  -  \mbox{ln det}( \bmF_{i_l} \bmR_{\overline{i_l}}^{ant} \bm{F}_{i_l}^H  )]
    \\& + \sum_{{i_r} \in \mathcal{R}} w_{i_r}[ \mbox{ln det}(\bm{F}_{i_r} \bmR_{i_r}^{ant} \bmF_{i_r}^H) - \mbox{ln det}(\bmF_{i_l} \bmR_{\overline{i_l}}^{ant} \bmF_{i_l}^H )] 
\end{aligned}.
\end{equation}
in which both the terms are purely concace, in contrast to \eqref{WSR_concave} in which the trace terms was only linear. To optimize the analog combiners, we take the derivative of \eqref{WSR_restated} with respect to $\bmF_{i_l}$
and $\bmF_{i_r}$, which leads to the following KKT conditions

\begin{subequations}
\begin{equation}
\begin{aligned}
 w_{i_l} & \bmR_{i_l}^{ant} \bmF_{i_l}^H  \big( \bm{F}_{i_l} \bmR_{i_l}^{ant} \bmF_{i_l}^H \big)^{-1}  \\& - w_{i_l} \bmR_{\overline{i_l}}^{ant} \bmF_{i_l}^H  \big( \bmF_{i_l} \bmR_{\overline{i_l}}^{ant} \bm{F}_{i_l}^H\big)^{-1} =0, 
\end{aligned}
\end{equation}
\begin{equation}
\begin{aligned}
 w_{i_r} & \bmR_{i_r}^{ant} \bm{F}_{i_r}^H  \big( \bm{F}_{i_r} \bmR_{i_r}^{ant} \bm{F}_{i_r}^H \big)^{-1}  \\& - w_{i_r} \bmR_{\overline{i_r}}^{ant} \bm{F}_{i_r}^H  \big( \bm{F}_{i_r} \bmR_{\overline{i_r}}^{ant} \bm{F}_{i_r}^H\big)^{-1} =0.
\end{aligned}
\end{equation}
\end{subequations}
Given the structure of the KKT conditions, it is immediate to see that the analog combiner can be optimized as a dominant generalized eigenvector solution of

\begin{equation} \label{comb_Fil}
    \bm{F}_{i_l} = \bmD_{N_{i_l}^r}(\bmR_{i_l}^{ant},\bmR_{\overline{i_l}}^{ant}), \quad \bm{F}_{i_r} = \bmD_{N_{i_r}^r}(\bmR_{i_r}^{ant},\bmR_{\overline{i_r}}^{ant}),
\end{equation}
from which we select $M_{i_l}^r$ and $M_{i_r}^r$ rows, respectively. Given the optimal analog combiner, operation $\angle\bm{F}_{i_l}$ and $\angle\bm{F}_{i_r}$ is required to meet the unit modulus constraint.

\subsection{Optimal power allocation}
Given the optimal beamformers, in this section we present a novel power allocation scheme for the massive MIMO FD interference channel.

Let $\bmP_{i_l}$ and $\bmP_{i_r}$ denote the power allocation matrix for the node $i_l \in \mathcal{L}$ and $i_r \in \mathcal{R}$, respectively. Let  $\bmSigma_{i_l}^{(1)}$,$\bmSigma_{i_l}^{(2)}$,$\bmSigma_{i_r}^{(1)}$ and $\bmSigma_{i_r}^{(2)}$ be defined as
\begin{subequations} \label{sigmas}
\begin{equation}
    \bmSigma_{i_l}^{(1)} = \bmV_{i_r}^H \bmG_{i_r}^H \bmH_{i_l,i_r}^H \bmF_{i_l}^H \bmR_{\overline{i_l}}^{-1}\bmF_{i_l} \bmH_{i_l,i_r} \bmG_{i_r}  \bmV_{i_r},
\end{equation}
\begin{equation}
    \bmSigma_{i_l}^{(2)} =  \bmV_{i_r}^H \bmG_{i_r}^H (\hat{\bmA}_{i_l} + \hat{\bmB}_{i_l} + \lambda_{i_r} \bmI) \bmG_{i_r}  \bmV_{i_r},
\end{equation}
\begin{equation}
    \bmSigma_{i_r}^{(1)} =  \bmV_{i_l}^H \bmG_{i_l}^H \bmH_{i_r,i_l}^H \bmF_{i_r}^H \bmR_{\overline{i_r}}^{-1}\bmF_{i_r} \bmH_{i_r,i_l} \bmG_{i_l}  \bmV_{i_l},
\end{equation}
\begin{equation}
    \bmSigma_{i_r}^{(2)} = \bmV_{i_l}^H \bm{G}_{i_l}^H  (\hat{\bmC}_j + \hat{\bmD}_j +  \lambda_{i_l} \bmI) \bmG_{i_l}  \bmV_{i_l}.
\end{equation}
\end{subequations}
Given \eqref{sigmas}, the problem \eqref{WSR_concave}
with respect to the power matrices can be stated as 
\begin{subequations}
\begin{equation} \label{power}
\begin{aligned}
& \underset{\substack{\bmP_{i_l}}}{\text{max}}  \sum_{i_l \in \mathcal{L}} w_{i_l} \mbox{ln det} (\bm{I} + \bmSigma_{i_l}^{(1)} \bmP_{i_l} )  - \mbox{Tr}( \bmSigma_{i_l}^{(2)} \bmP_{i_l}), \\& \underset{\substack{\bmP_{i_l}}}{\text{max}} \sum_{i_r \in \mathcal{R}} w_{i_r} \mbox{ln det} (\bmI + \bmSigma_{i_r}^{(1)} \bmP_{i_r} ) 
           - \mbox{Tr}( \bmSigma_{i_r}^{(2)} \bmP_{i_r} ).
\end{aligned}
\end{equation}
\end{subequations}

To include the optimal power allocation, we take the derivative of \eqref{power} for the power matrices $\bmP_{i_l}$ and $\bmP_{i_r}$, and solving  KKT conditions for the powers leads to the following optimal power allocation matrices

\begin{equation} \label{pow_il}
    \bmP_{i_l} = ( w_{i_l} {\bmSigma_{i_l}^{(2)}}^{-1} \hspace{-2mm} - {\bmSigma_{i_l}^{(1)}}^{-1} )^{+}, \quad
  \bmP_{i_r} = ( w_{i_r} {\bmSigma_{i_r}^{(2)}}^{-1} \hspace{-2mm} - {\bmSigma_{i_r}^{(1)}}^{-1} )^{+}.
\end{equation}
where $(x)^+ = \mbox{max}\{0,x\}$. To meet the sum-power constraint, we search the optimal multipliers satisfying the power constraints while updating the powers with \eqref{pow_il}.
The multipliers $\lambda_{i_l}$ and $\lambda_{i_r}$ should be such that the Lagrange dual function \eqref{WSR_concave} is finite and the values of the multipliers should be strictly positive. Let $\mathbf{\Lambda}$ and $\bmP$ denote the collection of multipliers and powers. Formally, the multipliers' search problem can be stated as
\begin{equation} \label{dual_func}
\begin{aligned}
  \underset{\mathbf{\Lambda}}{\text{min}}\,\, \underset{\bm{P}}{\text{max}}   \quad  \mathbb{L}\big(\mathbf{\Lambda},\bm{P}\big), 
  \quad \quad \mbox{s.t.}  \quad \mathbf{\Lambda} \succeq 0.
\end{aligned}
\end{equation}
The dual function $\underset{\bm{P}}{\text{max}} \;\mathbb{L}(\mathbf{\Lambda},\bm{P})$ is the pointwise supremum of a family of functions of $\mathbf{\Lambda}$, it is convex \cite{boyd2004convex}  and the
globally optimal values for $\mathbf{\Lambda}$ can be found by using any of the numerous convex optimization techniques. In this work,we adopt the bisection method. Let $\underline{\lambda_{i_a}}$ and $\overline{\lambda_{i_a}}$, for $i_a \in \mathcal{L} $ or   $i_a \in \mathcal{R} $, denote the upper and lower bound for the Lagrange multiplier search for node $i_a$. Let $[0, \overline{\lambda_{i_a}^{max}}]$ denote the maximum search interval for the multipliers for node $i_a$.
The WSR maximization hybrid beamforming design for the massive MIMO interference channel is given in Algorithm \ref{algo1}

\begin{algorithm}[!t]
\caption{Hybrid Beamforming and Combining}
\textbf{Given:} $\mbox{The CSI and rate weights.}$\\
\textbf{Initialize:} All the beamformers and combiners.\\
\textbf{Repeat until convergence} 
\begin{algorithmic}
\STATE \hspace{0.001cm}  for $\forall\; i_l \in \mathcal{L}$ and $\forall\; i_r \in \mathcal{R}$
\STATE \hspace{0.4cm} Compute $\bmA_{i_l}$ and  $\bmA_{i_l}$ with \eqref{A_il} and \eqref{A_ir}
\STATE \hspace{0.4cm} Compute $\bmB_{i_l}$ and  $\bmB_{i_r}$ with \eqref{B_il} and \eqref{B_ir}
\STATE \hspace{0.4cm} Compute $\bm{G}_{i_l}$ with
\eqref{analog_Gil}, do $\mbox{unvec}(\bm{G}_{i_l})$
and $\angle\bm{G}_{i_l}$
\STATE \hspace{0.4cm} Compute $\bm{G}_{i_r}$ with
\eqref{analog_Gir}, do $\mbox{unvec}(\bm{G}_{i_r})$
and $\angle\bm{G}_{i_r}$
\STATE \hspace{0.4cm} Compute $\bm{F}_{i_l}$, $\bm{F}_{i_r}$ with \eqref{comb_Fil} and do $\angle\bm{F}_{i_l}$, $\angle\bm{F}_{i_r}$
\STATE \hspace{0.4cm} Compute $\bm{V}_{i_l}$ and $\bm{V}_{i_r}$ with \eqref{digital_Vil} and \eqref{digital_Vir}
\STATE \hspace{0.4cm} Set $\underline{ \lambda_{i_l}} = 0$ and $\overline{\lambda_{i_l}} =\mu_{i_{max}}$ $\forall i_l \in \mathcal{L}_l$
\STATE \hspace{0.4cm} \textbf{Repeat until convergence} $\forall i_l \in \mathcal{L}$
\STATE \hspace{0.8cm} set $\lambda_{i_l} = (\underline{\lambda_{i_l}} + \overline{\lambda_{i_l}})/2$
\STATE \hspace{0.8cm} Compute $\bm{P}_{i_l}$ with \eqref{pow_il}, 
\STATE \hspace{0.8cm} \textbf{If} constraint for $\lambda_{i_l}$ is violated,
\STATE \hspace{1.1cm} set $\underline{\lambda_{i_l}} = \lambda_{i_l}$, \textbf{else} $\overline{\lambda_{i_l}} = \lambda_{i_l}$,
\STATE \hspace{0.4cm} \textbf{Repeat until convergence} $\forall i_r \in \mathcal{R}$
\STATE \hspace{0.8cm} set $\lambda_{i_r} = (\underline{\lambda_{i_r}} + \overline{\lambda_{i_r}})/2$. 
\STATE \hspace{0.8cm} Compute $\bm{P}_{i_l}$ with \eqref{pow_il}, 
\STATE \hspace{0.8cm} \textbf{If} constraint for $\lambda_{i_r}$ is violated,
\STATE \hspace{1.1cm} set $\underline{\lambda_{i_r}} = \lambda_{i_r}$, \textbf{else} $\overline{\lambda_{i_r}} = \lambda_{i_r}$,
\end{algorithmic}
\label{algo1}
\end{algorithm}
\subsection{Short Convergence Proof}
For the WSR cost function \eqref{WSR}, we construct its minorizer as in \eqref{WSR_concave}, which restates the WSR maximization as a concave problem, when the remaning variables are fixed (in the gradients. The minorizer results to be a touching lower bound for the original WSR problem, therefore we can write  
\begin{equation}
\begin{aligned}
        \underline{WSR} &= \sum_{i_l \in \mathcal{L}} w_{i_l} \mbox{ln det} (\bm{I} + \bmSigma_{i_l}^{(1)} \bmP_{i_l} ) - \mbox{Tr}( \bmSigma_{i_l}^{(2)} \bmP_{i_l}), \\& + \sum_{i_r \in \mathcal{R}} w_{i_r} \mbox{ln det} (\bmI + \bmSigma_{i_r}^{(1)} \bmP_{i_r} )
           - \mbox{Tr}( \bmSigma_{i_r}^{(2)} \bmP_{i_r} )
\end{aligned}
\end{equation}
The minorizer, which is concave in transmit covariance matrices, still has the same gradient of the original WSR and hence the KKT conditions are not affected. Now reparameterizing the transmit covariance matrices in terms of beamformer with the optimal power matrices and adding the power constraints to the minorizer, we get the Lagrangian \eqref{WSR_concave}. By updating all the beamformers and combiners as a dominant generalized  eigenvector, leads to an increase in the WSR \cite{kim2011optimal} at each iteration, thus assuring convergence.

\section{SIMULATION RESULTS}
This section presents the simulation results for the proposed HYBF and combining design for WSR maximization in a mmWave FD massive MIMO interference channel.

We consider a scenario of $2$ point-to-point communication links consisting of $4$ nodes operating in the mmWave. We assume that all the nodes are equipped with uniform linear arrays, with transmit antennas placed at half-wavelength. The angle of arrival $\phi_{i_l}^{n_p,n_c}$ and angle of departure $\theta_{i_l}^{n_p,n_c}$ are chosen to be uniformly distributed in the interval $[-20^{\circ},20^{\circ}]$. The SI channel is modelled with the Rician factor $\kappa_{i_l}= 10^{5}$ dB and the $r_{m,n}$ is modelled as in \cite{satyanarayana2018hybrid}, with distance between the transmit and receive array of $20$ cm and with a relative angle of $90^{\circ}$. The number of clusters and the number of paths is set to be $N_c = 3$ and $N_p = 6$. We define the transmit SNR at node $i_l$ as $\mbox{SNR}_{i_l} = p_{i_l}/\sigma_{i_l}^2$ and assume it to be the same for all the nodes, denoted as SNR.
We label the proposed hybrid beamforming scheme as HYBF, and for comparison purposes, we define the following benchmark schemes.
\begin{itemize}
    \item \emph{Fully Digital FD} - with all the FD nodes having number of RF chains equal to the number of antennas.
    \item \emph{Fully Digital HD} - with all the nodes having number of RF chains equal to the number of antennas but operating in half-duplex (HD) mode.
\end{itemize}
 
 Figure \ref{fig:my_label1} shows the achieved WSR as a function of SNR with a different number of RF chains in comparison with the benchmark schemes. It can be clearly seen that the proposed hybrid beamforming and combining has the potential to perform very close to the fully digital FD beamforming design with $32$ RF chains. Still the proposed algorithm exhibits some gap compared to the fully digital solution as the analog beamforming stage is constrained and must meet the unit modulus constraint. We can see that also with $16$RF chains, the proposed algorithm can considerably outperform the fully digital HD scheme, with operates with $100$ transmit and receive RF chains. Figure \ref{fig:my_label2} shows the average WSR as a function of SNR with $64$~transmit and receive antennas. It can be seen that with the same number of RF chains ($32$) and less number of transmit and receive antennas, the performance of the proposed HYB algorithm gets strictly close to the fully digital FD scheme.

  \begin{figure}
      \centering
      \includegraphics[width=8.5cm,height=7cm]{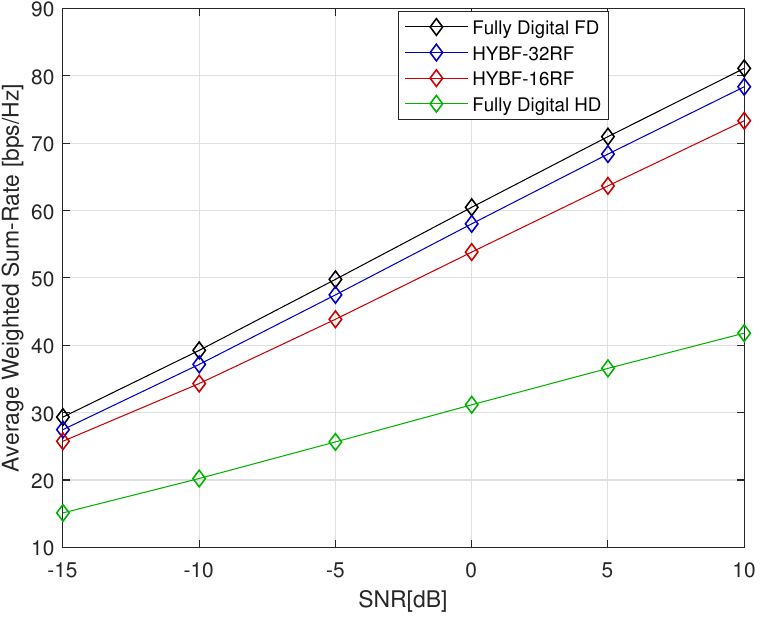}
      \caption{Average WSR as a function of SNR with transmit and receive antennas $N_{a}^t = N_{a}^r =100$ and RF chains $M_{a}^t = M_{a}^r = 32$ or $16$, $\forall a \in \mathcal{L}$ or $\mathcal{R}$.}
      \label{fig:my_label1}
  \end{figure}
  
  \begin{figure}
      \centering
      \includegraphics[width=8.5cm,height=7cm]{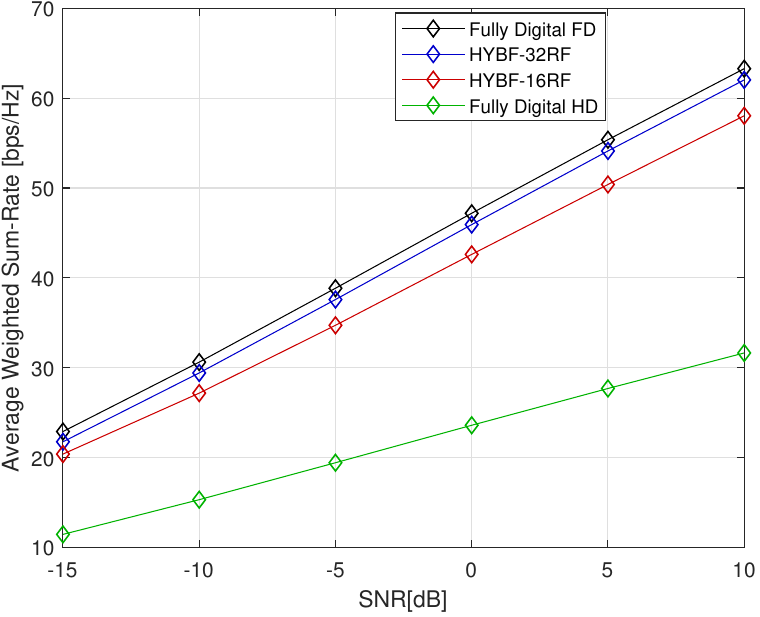}
      \caption{Average WSR as a function of SNR with transmit and receive antennas $N_{a}^t = N_{a}^r =64$ and RF chains $M_{a}^t = M_{a}^r = 32$ or $16$, $\forall a \in \mathcal{L}$ or $\mathcal{R}$.}
      \label{fig:my_label2}
  \end{figure}

\section{CONCLUSIONS}

In this paper, we studied the problem of WSR maximization in a mmWave massive MIMO FD interference channel consisting of $K$ point-to-point communication links. The simultaneous coexistence of multiple FD nodes leads to an extremely challenging communication scenario for which a novel hybrid beamforming and combining scheme is proposed. Simulation results show that the proposed design can perform extremely close to the fully digital FD beamforming design operating with $100$ antennas, with only $32$ RF chains. Moreover, the proposed design significantly outperforms the fully digital HD system with only $16$ RF chains.

\bibliographystyle{IEEEtran}
\bibliography{Globcomm_2021.bib}

\end{document}